\documentclass{article}
\usepackage{graphicx} 
\usepackage{geometry}
\usepackage{listings}
\usepackage{xcolor}
\usepackage{algorithm}
\usepackage{algpseudocode}

\usepackage{amsfonts, amsmath, amssymb, amsthm}
\usepackage{thmtools}
\usepackage{hyperref}
\usepackage{cleveref}

\declaretheorem[name=Theorem,numberwithin=section]{theorem}

\newtheorem{lemma}{Lemma}[section]

\newtheorem{corollary}[theorem]{Corollary}

\crefname{conjecture}{Conjecture}{Conjectures}

\theoremstyle{definition}

\newcommand{\istrut}[2][0]{\rule[- #1 mm]{0mm}{#1 mm}\rule{0mm}{#2 mm}}

\newcommand{\Szemeredi}{Szemer\'{e}di}
\newcommand{\Patrascu}{P{\v{a}}tra\c{s}cu}

\newcommand{\Lovasz}{Lov\'{a}sz}

\newcommand{\Search}{\textsf{Search}}
\newcommand{\Insert}{\textsf{Insert}}
\newcommand{\Delete}{\textsf{Delete}}
\newcommand{\Index}{\textsf{Index}}
\newcommand{\Pos}{\textsf{Pos}}
\newcommand{\Subtree}{\textsf{Subtree}}

\newcommand{\bydef}{\stackrel{\operatorname{def}}{=}}

\newcommand{\child}{\operatorname{child}}

\title{A Simple Analysis of Quadratic Probing\\ and Other Open Addressing Schemes\thanks{This work was conducted while the first and third authors were visiting University of Michigan, Ann Arbor.  Supported by NSF Grant CCF-2446604.}}

\author{Yuhao Guo\\{\small IIIS, Tsinghua University}
\and 
Seth Pettie\\
{\small University of Michigan}
\and
Chengzhang Wan\\
{\small IIIS, Tsinghua University}
}

\date{}

\begin{document}

\maketitle

\begin{abstract}
    In open addressed hashing, \emph{quadratic probing} is attractive for striking a nice balance between having a high locality of reference and a low number of probes per search.  
    However, these are \emph{empirical observations}, 
    not theoretical guarantees.  Indeed, until recently, it was not known whether quadratic probing had constant expected insertion cost under \emph{any} positive load factor $\alpha > 0$, even with uniformly random hash functions.

\medskip 

    In a recent breakthrough---albeit a numerically understated breakthrough---Kuszmaul and Xi~\cite{KuszmaulX24} proved 
    that any fixed offset sequence (including quadratic probing) does, 
    in fact, have constant expected insertion cost for load factors $\alpha \leq 8.9\%$.
    This is well below what we would like to prove, that quadratic probing has constant insertion cost for any load factor 
    $\alpha < 1-\epsilon$ bounded away from 1.

\medskip 

    In this paper, we prove that open addressed hashing with 
    any fixed offset sequence has constant expected insertion cost 
    for load factors up to $35.74\%$, and that for quadratic probing in particular, we can increase the load factor to $37.61\%$.  Our main innovation is 
    a new type of \emph{witness forest} 
    for recording collisions among the probe sequences.
\end{abstract}

\section{Introduction}

An \emph{open addressed hash table} is 
one that stores a dynamic \emph{dictionary} $S\subset [U]$ 
in a flat array with $n$ slots, where a newly inserted element $x\in [U]$ is put in the first available slot indexed by:
\[
h(x,0), h(x,1), h(x,2), \ldots 
\]
Such a scheme supports \emph{searches} and \emph{deletions} (via tombstones~\cite{BenderKK21}) 
in the same time as insertions.
The main design question is: what is the best probe sequence $(h(x,i))_{i\geq 0}$, taking into account factors such as 
locality of reference/cache misses, overhead for hash function evaluation, and quality of the hash functions needed, among 
other factors.  
By default we shall assume that all hash functions are uniformly random functions. 
The canonical probe sequences are
\begin{description}
    \item[Uniform Hashing.] Here $h : [U]\times \mathbb{Z} \to [n]$ is uniformly random, or equivalently,
    there are independent hash functions 
    $h_0,h_1,h_2\ldots : [U]\to[n]$, and 
    \[
    h(x,i) = h_i(x).
    \]

    \item[Double Hashing.] There are two independent 
    hash functions $h_0: [U]\to [n], h_1 : [U] \to [n]-\{0\}$, and
    \[
    h(x,i) = (h_0(x) + i\cdot h_1(x)) \text{ mod } n.
    \]
    
    \item[Linear Probing.] There is a single hash function $h : [U]\to [n]$, and
    \[
    h(x,i) = (h(x) + i) \text{ mod } n.
    \]
    That is, we search sequentially, starting at the random slot 
    $h(x)$.
    \item[Quadratic Probing.] Like linear probing there is one hash function $h : [U]\to [n]$, but the probe sequence is less clustered.
    \[
    h(x,i) = (h(x) + i^2) \text{ mod } n.
    \]
    Some texts use \emph{quadratic probing} to refer to 
    a family of probe sequences parameterized by constants $a,b$.
    \[
    h(x,i) = (h(x) + ai + bi^2) \text{ mod } n.
    \]
    
    \item[Fixed-offset Probing.] There is a single hash function $h : [U]\to [n]$ and a fixed permutation 
    $r_0=0$, $r_1$, $r_2$, $r_3,\ldots,r_{n-1}$ of $[n]$.
    The probe sequence is then
    \[
    h(x,i) = (h(x) + r_i) \text{ mod } n.
    \]
    Fixed-offset probing subsumes linear probing, 
    and \emph{basically} subsumes quadratic probing.  
    The issue is that 
    $(i^2 \text{ mod } n)_{0\leq i\leq n-1}$ and 
    $(ai + bi^2 \text{ mod } n)_{0\leq i\leq n-1}$ 
    are not necessarily permutations
    of $[n]$, though their first $\Omega(\sqrt{n})$ terms are 
    distinct, which is often just as good.\footnote{See Hopgood and Davenport~\cite{HopgoodD72} and Batagelj~\cite{Batagelj75} for the design of quadratic probing sequences for which the first $n/2$ or $n-O(1)$ or $n$ probes are guaranteed to be distinct.}
\end{description}

Linear probing was proposed in the 1950s~\cite{Peterson57} and famously analyzed by Knuth~\cite{Knuth63}, Schay and Spruth~\cite{SchayS62}, and Konheim and Weiss~\cite{KonheimW66} in the 1960s.  These and later analyses~\cite{MendelsonY80,FlajoletPV98,Janson01} assume uniformly random hash functions.  With load factor $1-\epsilon$,
the expected number of probes for the next insertion
is $\frac{1}{2}(1+\epsilon^{-2})$~\cite{Knuth63},
which is much worse than the $\epsilon^{-1}$ of uniform 
probing.  Nonetheless, linear probing is popular in practice due to its simplicity and high locality of reference.
Pagh, Pagh, and Ruzic~\cite{PaghPR09} proved that the expected $O(\epsilon^{-2})$ bound for every insert and 
query could be achieved with $5$-wise independence,
and \Patrascu{} and Thorup~\cite{PatrascuT16} proved that 4-wise independence is insufficient.\footnote{In particular, 
4-wise independence, but not 3-wise independence, suffices to guarantee that 
$(1-\epsilon)n$ insertions take 
$O_\epsilon(n)$ time in total~\cite{PaghPR09}, 
but to get an $O(\epsilon^{-2})$ bound 
for each insertion individually,
5-wise independence but not 4-wise independence suffices~\cite{PatrascuT16}.
The upper bounds hold for \emph{any} 
$k$-wise independent hash family, and the lower bounds hold for a \emph{worst case} $k$-wise independent hash family.  Some non-5-wise independent classes of hash 
functions~\cite{PatrascuT12,PatrascuT13,Bercea0KHT23} 
still guarantee $O(\epsilon^{-2})$ expected insertion time.}

\medskip 

Double hashing was proposed by Balbine~\cite{Balbine68} 
and Bell and Kaman~\cite{BellK70}.  
Guibas and \Szemeredi~\cite{GuibasS76,GuibasS78} 
showed that it has constant expected insertion time 
for load factors up to
$\alpha \leq 0.319$~\cite{GuibasS76,GuibasS78}, 
and Lueker and Molodowitch~\cite{LuekerM93} eventually 
proved that the same holds for all load factors $\alpha = 1-\epsilon$ 
bounded away from 1.

\medskip 

Quadratic probing was proposed by Maurer~\cite{Maurer68} 
in the 1960s.  
It combines some of the good features
of linear probing, namely locality of reference and hence few cache misses, as well as 
uniform/double hashing, 
which create less clustering due to 
their more scattered probe sequences.
Despite the long history of hashing 
and algorithmic analysis, 
essentially nothing was known about quadratic probing \emph{theoretically} until very recently.
More than 55 years after it was proposed, 
Kuszmaul and Xi~\cite{KuszmaulX24} finally 
proved that under any load factor  
$\alpha \leq 0.089$ ($8.9\%$ of the table), 
quadratic probing, and in fact any fixed-offset 
probe sequence,
has constant expected insertion 
time.\footnote{Specifically, 
Kuszmaul and Xi~\cite{KuszmaulX24} 
proved that any fixed probe sequence 
has constant expected insertion time
for any load factor
$\alpha < \alpha_{\text{KX}}$, 
where 
$\beta_{\text{KX}}\alpha_{\text{KX}}e^{1-\alpha_{\text{KX}}}=1$
and $\beta_{\text{KX}} \approx 4.51$ is the minimum value
attained by 
$f(u)=u^{-1}\prod_{i\geq 1}(1+u^i)$.}
They also showed that the blocked variant of 
quadratic probing---a hybrid between it and linear 
probing---can achieve load factors close to 1 for suitable large block size.\footnote{Several modern libraries use blocked variants of quadratic probing, e.g., 
\url{https://abseil.io}
and
\url{https://docs.rs/hashbrown/latest/hashbrown/}.}

\subsection{Our results}

The Kuszmaul-Xi~\cite{KuszmaulX24} analysis is based on analyzing \emph{witness strings}, which encode information about the collision pattern of elements as they proceed along their 
probe sequences.  However, they are not as efficient as they could be.

\medskip 

In this paper we introduce a simple \emph{witness forest} 
that efficiently encodes similar information. 
Applying this definition to an arbitrary fixed-offset 
sequence $(r_0,r_1,r_2,\ldots)$, we prove that constant expected insertion time is achievable for any load $\alpha \leq 35.74\%$.

\begin{theorem}\label{thm:mainthm}
Define $\alpha^* \approx 0.357403$ to be the unique solution of $\frac{4}{e}\alpha e^{1-\alpha}=1$ in the range $[0,1]$. 
For any load factor $\alpha < \alpha^*$, 
and any permutation $(r_0,\ldots,r_{n-1})$,
the cost of inserting a new element is dominated by a geometric random variable with mean $O(1)$.
\end{theorem}

\Cref{cor:quadprobing} only relies on the fact that the first $\Omega(\sqrt{n})$ elements of the probe sequence are distinct mod $n$.  The $\exp(-\Omega(\sqrt{n}))$ failure probability can be removed by choosing $n$ and the probe sequence more carefully; see~\cite{HopgoodD72,Batagelj75}.

\begin{corollary}\label{cor:quadprobing}
For any load factor $\alpha < \alpha^*$, with probability $1-\exp(-\Omega(\sqrt{n}))$, quadratic probing supports insertions in expected constant time.
\end{corollary}

The $\frac{4}{e}\alpha e^{1-\alpha}$ expression comes from an enumeration of labeled trees that includes all realizable witness trees, but many trees that are \emph{unrealizable} for the specific offset sequence $(r_0,r_1,\ldots)$.
We develop a method using generating functions to give arbitrarily good upper bounds
on the number of \emph{realizable} witness trees for any 
particular fixed-offset sequence $(r_0,r_1,r_2,\ldots)$.
Applied to basic quadratic probing, where $r_i=i^2$,
this method shows that~\Cref{thm:mainthm} and \Cref{cor:quadprobing} can go slightly higher, 
to load factor $37.61\%$. 
The use of generating functions is not new in the analysis
of linear probing~\cite{MendelsonY80,FlajoletPV98,Janson01}
or even quadratic probing~\cite{KuszmaulX24},
but our use of generating functions is more elementary.
We combine exact counting of small realizable 
trees and a general recurrence for larger trees 
to yield substantially better upper bounds on the 
number of realizable trees.

\subsection{Related Results}

The first sentence of this paper defined ``open addressing'' in the most restrictive way possible.  
In some papers open addressing simply means that the elements of $[U]$ are stored in an array, with no restrictions on the insertion, deletion, or search algorithms.  Other papers start with the restrictive definition, but explore the consequences of 
explicitly dropping one or more of its axioms.  
For example, 
\begin{itemize}
    \item What if the insert algorithm is \emph{not obliged} to put $x\in [U]$ in the \emph{first} available slot in its probe sequence?
    \item What if we can rearrange elements after insertion?
    \item What if we can add tombstones~\cite{BenderKK21}
    that \emph{do not} correspond to deleted elements?

    \item What if there is a possibility of failure, where there is no feasible assignment of elements to array locations?
\end{itemize}

Because quadratic probing is so poorly understood, we focus only on expected insertion time, which is an upper bound on expected search time.  However, most prior work
focuses on insertion, deletion, and search times, 
under worst case, worst case expected, 
and average expected measures, 
and explores tradeoffs between these measures, 
sometimes differentiating between successful searches ($x\in S$) and unsuccessful searches ($x\not\in S$).
A good starting place for discussing these fine distinctions is Yao's~\cite{Yao85} lower bound from 1985.

\medskip 

Uniform hashing guarantees that the average time of 
a successful search, over all $x\in S$, is $O(\log\epsilon^{-1})$ under load factor 
$\alpha = 1-\epsilon$,\footnote{We shall assume for simplicity that $\epsilon \leq 1/2$.} 
Yao~\cite{Yao85} proved that no open addressed scheme in the restrictive sense can do better, that is, if an insert places the element in the first available slot in its probe sequence, and never moves it thereafter.

Farach-Colton, Krapivin, and Kuszmaul~\cite{Farach-ColtonKK24} illustrated the fragility of Yao's lower bound to the modeling assumptions.  
First, they gave a scheme (\emph{elastic hashing}) 
with $O(1)$ average query time and $O(\log\epsilon^{-1})$ worst case expected query time
in which elements are not necessarily inserted in the first available slot.  Second, they showed that
even in the restricted open addressing model, the worst case expected query time of uniform hashing could be improved, from $O(\epsilon^{-1})$ to $O(\log^2\epsilon^{-1})$ with \emph{funnel hashing}, 
which inserts every element into the first available slot in its probe sequence.

The idea of rearranging elements after 
insertion goes back at least to 
the \emph{ordered hash tables} of 
Amble and Knuth~\cite{AmbleK74} 
and \emph{Robin Hood hashing} of 
Celis, Larson, and Munro~\cite{CelisLM85}.  
When attempting to insert $x$ into position 
$l = h(x,i)$, if the element $y$ at position $l$
has $l = h(y,j)$ then we may wish to swap $x$ and $y$
and continue with the next position in $y$'s probe sequence $h(y,j+1)$.  The \emph{Robin Hood} rule~\cite{CelisLM85} swaps $x$ and $y$ when $i>j$ 
while the \emph{anti-Robin Hood} rule~\cite{HuKLWYZ26} swaps when $i<j$.\footnote{Ties can be broken arbitrarily or via some consistent rule.  A different formulation of the Robin Hood rule~\cite{HuKLWYZ26}, which often has the same outcome, is that we should prioritize $x$ or $y$ according to the distance from $l$ to $h(x)$ or $h(y)$.}
In linear probing, the Robin Hood rule does not change the 
average search time, but it does reduce its variance, balancing the search times of elements inserted early and late.
Using the full power of rearranging elements, 
Bender, Kuszmaul, and Zhou~\cite{BenderKZ24}
gave optimal bounds under open addressing with 
probe sequences $h(x,i)$.  At load factor $\alpha=1-\epsilon$, their scheme has $O(1)$ expected search time,
and $O(\log\log\epsilon^{-1})$ expected time for insertion and deletion, which is optimal under these constraints.  
It even works at load factor $1$, with $O(\log\log n)$ expected insertion/deletion time.

Bender, Kuszmaul, and Kuszmaul~\cite{BenderKK21} 
demonstrated that Knuth's $\Theta(\epsilon^{-2})$
expected insertion time at load $\alpha = 1-\epsilon$
is perhaps too pessimistic.  If the hash table
is subject to a long sequence of insertions and deletions, but never exceeds load factor $1-\epsilon$,
then \emph{graveyard hashing}~\cite{BenderKK21}, a small modification to linear probing, achieves expected time $\tilde{O}(\epsilon^{-1})$ per operation.  This was later refined to $\Theta(\epsilon^{-1}\log^{3/2}\epsilon^{-1})$
by Braverman and Kuszmaul~\cite{BravermanK24}.
Graveyard hashing implements deletes with \emph{tombstones}, which have an anti-clustering effect,
and occasionally rebuilds the table, inserting artificial tombstones into long runs.

Very recently Hu, Kuszmaul, Liang, Walzer, Yu, and Zhou~\cite{HuKLWYZ26} analyzed ``smoothed'' variants
of quadratic probing. Rather than have a single fixed-offset sequence $(r_0,r_1,\ldots)$, each element $x\in [U]$ is assigned a sorted sequence $(r_0(x),r_1(x),r_2(x),\ldots)$ 
by independently including each integer $\ell$ in the sequence with probability $1/\sqrt{\ell}$.
They prove that if elements are shifted 
according to a Robin Hood rule, the expected
insertion and query times are $\Theta(\epsilon^{-3/2})$
and $O(\epsilon^{-1/2})$, respectively,
but under an anti-Robin Hood rule, the expected time 
bounds improve to $\Theta(\epsilon^{-1})$ and $O(\log\epsilon^{-1})$, respectively.\footnote{Note that in the case of linear probing, the set of occupied table locations is invariant under any Robin Hood or anti-Robin Hood type reordering rule.  Under a quadratic probing-type regime, however, the set of occupied locations can look quite different using first-come first-served, Robin Hood, or anti-Robin Hood rules.}  
They also prove that with a single random quadratic 
probing type offset sequence $(r_0,r_1,r_2,\ldots)$, 
where $\mathbb{E}(r_i) \approx i^2$, and an anti-Robin Hood rule,
the expected time for insert and search are $O(\epsilon^{-1})$ and $O(\log\epsilon^{-1})$, 
respectively.  
These results do not imply anything about 
the ``standard'' quadratic probing sequence 
$(r_i = i^2)_{i\geq 0}$, nor any quadratic probing-type discipline without any (anti-)Robin Hood-type rule.

\medskip 
Cuckoo hashing~\cite{PaghR04} 
can be considered an open addressing scheme, where the probe sequence of every element has length 2, and elements can be moved during inserts.  The downsides of cuckoo hashing are (1) the possibility of no feasible assignment of elements to array locations, 
and (2) a maximum load factor of $1/2-\epsilon$.  
Downside (1) has been mitigated by extending every probe sequence with a constant sized ``stash'' of shared locations; 
see Kirsch, Mitzenmacher, and Wieder~\cite{KirschMW09}.
Downside (2) has been mitigated by allowing more elements per cell and/or expanding the number of cells where an element is allowed to be placed; see~\cite{FotakisPSS05,DietzfelbingerW07,CainSW07,FernholzR07,LehmanP09,DietzfelbingerGMMPR10,DietzfelbingerMR11,FriezeM12,Lelarge12,FountoulakisP12,FountoulakisPS13,FountoulakisKP16,MitzenmacherPW18,FriezeJ19,Walzer22,Walzer23,BellF24,Walzer25,KuszmaulM25}.  Some of these works also mitigate downside (1).

\medskip 

Witness trees and other witness structures
have been used for a variety of probabilistic analyses.  A small sample outside of hashing 
includes witness trees for the constructive \Lovasz{} local lemma~\cite{MoserT10,ChungPS17}, 
witness trees for power-of-$d$-choices
allocation problems~\cite{SchickingerS00},
witness trees for online coloring avoiding monochromatic subgraphs~\cite{MutzeRS14}, and witness trees
for \emph{graph shattering}-type parallel 
symmetry breaking algorithms~\cite{RubinfeldTVX11,BarenboimEPS16}.
Within hashing we see witness strings
for quadratic probing~\cite{KuszmaulX24},
and witness graphs in the analysis
of cuckoo hashing~\cite{PaghR04}.

\subsection{Organization}

\Cref{sect:prelim} formally defines the 
\Insert, \Delete, and \Search{} operations
in hash tables using a 
fixed-offset probe sequence.
\Cref{sect:witness} defines the 
\emph{witness forest} and uses it to analyze
any fixed-offset probe sequence, leading
to results for load factors up to $35.74\%$.
\Cref{sect:quadratic-probing} gives
a method for improving this bound
for any \emph{specific} fixed-offset sequence,
and shows that quadratic probing in particular
($(r_i=i^2)_{i\geq 0}$) works up to load factor 
$37.61\%$.  
Appendix \ref{sect:linear-probing} applies 
the witness tree framework to prove that linear probing has expected constant time operations for any load factor $\alpha = 1-\epsilon$ bounded away from 1.

\medskip 

We conclude in \Cref{sect:conclusion} 
with some remarks.

\section{Preliminaries}\label{sect:prelim}

Let $A = A[0],\ldots,A[n-1]$ be the hash table, 
$h : U\rightarrow [n]$ be a uniformly random hash function, 
and $r_0,r_1,\ldots,r_{n-1}$ be a permutation of $\{0,1,\ldots,n-1\}$.
Without loss of generality 
assume $r_0=0$.
Every slot in $A$ may contain an element of $U$, 
a \emph{tombstone} ($\blacksquare$), or be empty ($\perp$).
\Insert, \Delete, and \Search{} are handled as follows:
\begin{description}
    \item[\Insert$(x)$ :] Find the minimum $i$ such that 
    $A[h(x)+r_i \text{ mod } n]\in\{\perp,\blacksquare\}$, and then set $A[h(x)+r_i \text{ mod } n]\gets x$.
    \item[\Delete$(x)$ :] Find the index $i$ such that $A[h(x) + r_i \text{ mod } n] = x$, then  
    set $A[h(x) + r_i \text{ mod } n]\gets \blacksquare$.
    \item[\Search$(x)$ :] Find the minimum $i$ such that $A[h(x)+r_i \text{ mod } n] \in \{x,\perp\}$.  Return $A[h(x)+r_i \text{ mod } n]$.
\end{description}
In other words, tombstones behave like elements in a \Search{}
or \Delete{},
and like empty cells in an \Insert.

We consider the situation after $m=\alpha n$ insertions have been performed, and analyze
the expected time for the $(m+1)$th insertion.
Without loss of generality, we may assume the sequence of key insertions 
$x_1,\ldots,x_{m+1}$ is just $1,2,\ldots,m+1$,
i.e., $x_i=i$.
For $x\in V \bydef \{1,\ldots,m+1\}$,
define $\Index(x)$ to be such that $A[h(x) + r_{\Index(x)} \text{ mod } n]=x$ and $\Pos(x) = h(x)+r_{\Index(x)} \text{ mod } n$.  
The cost of \Insert$(x)$ is $\Index(x)+1$.

\section{The Witness Forest}\label{sect:witness}

We construct a \emph{witness forest} that represents 
the sequence of probes made when inserting keys $1,2,\ldots,m+1$.  The witness forest $F$ consists
of rooted, ordered trees. 
Each node in $F$ is 
identified with an element of $V$, 
and also assigned a \emph{label} of an 
element, or left unlabeled ($\bot$).  
Let $\ell : V\to V\cup\{\bot\}$ 
be the labeling function.

The witness forest is constructed as follows.  
(See \Cref{fig:witness-forest} 
for one execution of this algorithm on a small example.)
Upon \Insert$(u)$ we
\begin{itemize}
\item Create a new node identified with $u$ and label $\ell(u) = \bot$. If $\Pos(u) = h(u)$, that is, $u$ was placed directly in the first location in its probe sequence, then $u$ is a singleton tree in $F$.

\item Otherwise, $\Index(u) > 0$. 
For $j$ from $0$ to $\Index(u)-1$,
\begin{itemize}
\item Let $v$ be the element in position $\Pos(v) = h(u)+r_j \text{ mod } n$, and $w$ be the root of $v$'s tree in $F$.
\item If $v$ is already a descendant of $u$ then do nothing.
\item Otherwise, make $w$ the rightmost child of $u$, 
and assign it the label $\ell(w) \gets v$.
\end{itemize}
\end{itemize}

\begin{figure}
    \centering
    \includegraphics[width=1\linewidth]{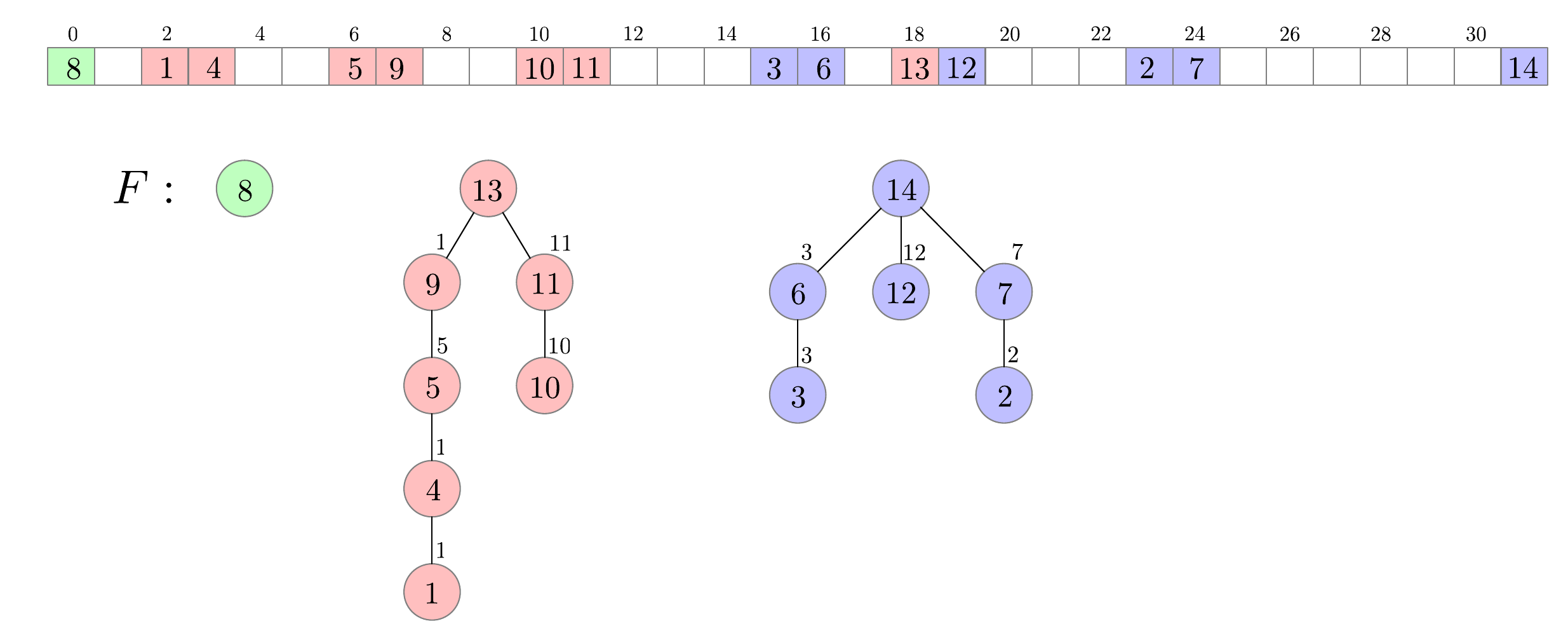}
    \caption{Top: the hash table $A[0..31]$ after inserting $1,2,\ldots,14$ using quadratic probing, where $r_i = i^2$.  Bottom: the corresponding witness forest,
    where the label $\ell(u)$ is written just above the node $u$.  The hash values in this example are:\\
    $\begin{aligned}
    \istrut{5}\qquad\qquad h(8) &= 0, & \qquad h(1)=h(4)=h(5)=h(13) &=2,     & h(2)=h(7) &= 23,\\
    &&h(9) &= 6, &     \qquad h(3)=h(6)=h(14) &= 15,\\
    &&h(10)=h(11) &= 10,    & h(12) &= 19.\\
    \end{aligned}$
    }
    \label{fig:witness-forest}
\end{figure}

\subsection{Structural Properties of Witness Trees}

Define $\Subtree(x,F)$ to be the subtree of $x$ in forest $F$.
We will drop $F$ when it is known from context, and also use
$\Subtree(x,F)$ to refer to the set of 
elements in $V$ identified with the 
vertices in the subtree of $x$.

\begin{lemma}\label{lem:forest-properties}
Let $F,\ell$ be the witness forest and labeling function 
constructed after inserting
elements $1,2,\ldots,m+1$ in that order.

\begin{enumerate}
    \item For any $u,v\in V$, 
    if $u$ is a strict ancestor 
    of $v$, then $u>v$.\label{item:ancestor-property}
    \item A node $u$ is unlabeled ($\ell(u)=\bot$) if and only if $u$ is a root in $F$.  If $\ell(u)\neq \bot$ then $\ell(u)\in \Subtree(u)$.\label{item:label-property}
    \item For any $j\in [0,\Index(u))$, 
    if $v\in V$ is at position $\Pos(v) = h(u) + r_j \;\mathrm{ mod }\; n$, then $v\in \Subtree(u)$.
    As a consequence, 
    $\Index(u) < |\Subtree(u)|$.
    \label{item:collision-property}
\end{enumerate}
\end{lemma}

\begin{proof}
    Claims (1), (2), and (3) follow directly from the construction.
\end{proof}

\Cref{cor:hashrange} is a straightforward consequence of \Cref{lem:forest-properties}.

\begin{corollary}\label{cor:hashrange}
Let $u,v,w \in V$ be elements such that
$w\in \Subtree(u)$ 
and $\Pos(w) = h(v)$.
Then $u$ and $v$ are related, 
that is,
either $v\in \Subtree(u)$ 
or $u\in \Subtree(v)$.
\end{corollary}

\begin{proof}
If $w=v$ then the claim holds trivially.  
If $w\neq v$ then $w<v$
and upon $\Insert(v)$, 
$w$'s subtree will be made a child of $v$
in $F$.  Thus, $u$ and $v$ are both 
ancestral to $w$, and hence related.
\end{proof}

\Cref{lem:relativepos} shows that 
given a tree $T$ in a witness forest,
we can recover the \emph{relative} 
hash values for all keys in $T$,
as well as their relative positions 
in the hash table $A$.

\begin{lemma}\label{lem:relativepos}
Any labeled tree $T$ in the witness forest $F$ 
uniquely determines $\Index(u)$ for all $u\in T$
and the \emph{relative} positions 
$h(u)-h(v)$, $\Pos(u)-\Pos(v)$, and 
$h(u)-\Pos(v)$ (modulo $n$) for all $u,v\in T$.
\end{lemma}

\begin{proof}
We prove the claim by induction.  
In the base case $T$ 
is a single-node tree
containing $u$, $\Index(u)=0$, 
and all the relative quantities are 0.

In the general case, let the root of $T$ be $u$,
and $u_0,\ldots,u_k$ be its children in left-to-right order.
By induction, the relative $h$-values and $\Pos$-values
are known for any $i\in [0,k]$ and $u',v'\in \Subtree(u_i)$. 
Now let us simulate the operation of \Insert$(u)$.
We probed $A[h(u)]=A[h(u)+r_0]$ and found that it is occupied.  By construction of the witness forest, $h(u) = \Pos(\ell(u_0))$, so now we know the relative position of all elements in $\{u\}\cup \Subtree(u_0)$,
and can ascertain whether 
$h(u)+r_j = \Pos(v')$ for any $j$ and $v'\in \Subtree(u_0)$.
Let $j_0$ be the first index for which
$h(u)+r_{j_0}$ does \emph{not} 
collide with any element of $\Subtree(u_0)$.\footnote{For example, in \Cref{fig:witness-forest},
if we are looking at the tree rooted at $u=13$,
then $j_0 = 3$ since $h(u)+0^2, h(u)+1^2, h(u)+2^2$
are all in the subtree rooted at $u_0 = 9$.
In the case of $u=14$, $j_0 = 2$ since only $h(u)+0^2, h(u)+1^2$ are in the subtree rooted at $u_0 = 6$.}
Thus, it must be the case
that $\Pos(\ell(u_1)) = h(u) + r_{j_0}$, 
and now we know the relative position of all elements in 
$\{u\}\cup \Subtree(u_0)\cup\Subtree(u_1)$.
In this way we can define $j_1,\ldots,j_k$
and identify the collisions at positions
$h(u) + r_{j}$ for all $j<j_k$.  
Since $u_k$ is the last child, it must be that $A[h(u)+r_{j_k}]$ was unoccupied at the time of \Insert$(u)$, 
hence
$\Index(u)=j_k$ and 
$\Pos(u) = h(u)+r_{j_k}$.
\end{proof}

\subsection{Counting Realizable Witness Trees}\label{sect:counting-realizable-witness-trees}

Call a labeled tree $T$ \emph{realizable} 
if it is possible to see $T$ in a witness forest after inserting elements $1,2,\ldots,m+1$.  
Whether a particular tree is realizable depends on the structure of the probe offset sequence $(r_0,r_1,\ldots,r_{n-1})$.
See \Cref{fig:realizable} for some examples of unrealizable witness trees.

\begin{figure}
    \centering
    \includegraphics[width=0.5\linewidth]{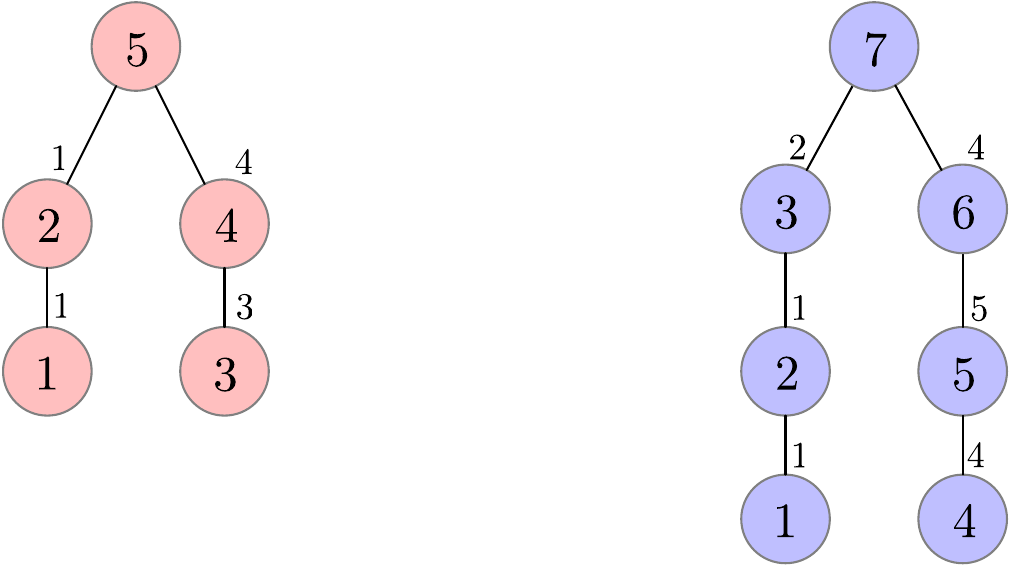}
    \caption{Left: a witness tree that is realizable for quadratic probing, but not linear probing.  Right: a witness tree that is realizable for linear probing, but not quadratic probing.\label{fig:realizable}}
\end{figure}

\begin{lemma}\label{lem:numgenerable}
Fix any subset $V'\subseteq V$ of the keys 
with size $s=|V'|$.  The number of realizable
labeled trees on key-set $V'$ is at most 
$4^{s-1}(s-1)!/s^{3/2}$.
\end{lemma}

\begin{proof}
Fix an \emph{unlabeled} tree topology $\tau$ 
on $s$ nodes.  
We will count the number of ways to identify
the nodes of $\tau$ with $V'$ that satisfy \Cref{lem:forest-properties}(1), 
then count the number of labelings that satisfy \Cref{lem:forest-properties}(2).  The resulting count will be an overestimate as it includes unrealizable trees for the particular offset sequence $(r_0,\ldots,r_{n-1})$.

Imagine assigning the node identities uniformly at random
from all $s!$ assignments.
An assignment satisfies \Cref{lem:forest-properties} if every node has the maximum value in its subtree, so the probability a random assignment is legal according to \Cref{lem:forest-properties}(1) is $\prod_{v\in \tau} \frac{1}{|\Subtree(v,\tau)|}$.
By \Cref{lem:forest-properties}(2) the label of a non-root must come from its subtree, 
so there are at most 
$s^{-1}\prod_{v\in \tau} |\Subtree(v,\tau)|$ labelings.
Thus, for a fixed tree topology $\tau$ we encounter a pleasant cancellation.  There are
\[
s! \cdot \prod_{v\in \tau} \frac{1}{|\Subtree(v,\tau)|} \cdot \frac{1}{s}\cdot \prod_{v\in \tau} |\Subtree(v,\tau)| = (s-1)!
\]
ways to assign the identities and labels from $V'$.
The number of such tree topologies is 
given by the 
Catalan number 
$C_{s-1} \leq 4^{s-1}/s^{3/2}$.
\end{proof}

\begin{lemma}\label{lem:numseq}
Fix any labeled tree $T$ on element set $V'\subseteq V$, where $|V'|=s$.
Let $\mathcal{H}$ be the set of all
possible hash vectors $(h(1),\ldots,h(m+1))$,
where $|\mathcal{H}|=n^{m+1}$.
At most $n(n-s)^{m+1-s}$ elements of 
$\mathcal{H}$ cause $T$ to appear in the witness forest $F$.
\end{lemma}

\begin{proof}
Let $u$ be the root of $T$.  There are $n$ choices for $h(u)$, and by \Cref{lem:relativepos}, 
once $h(u)$ is fixed, 
we know $h(v')$ for all $v'\in V'$.
Moreover, by 
\Cref{cor:hashrange},
for every $v\not\in V'$ we must have
$h(v) \not\in \Pos(V')$.
Thus, there are $(n-s)^{m+1-s}$ choices for the hash values of $V\backslash V'$.
\end{proof}

We are now in a position to prove \Cref{thm:mainthm}.  For $h\in\mathcal{H}$,
let $\Index(x,h)$ and $\Subtree(x,F(h))$
be $\Index(x)$ and $\Subtree(x)$ 
when $h$ is the hash function.

\begin{proof}[Proof of \Cref{thm:mainthm}]
The expected cost of \Insert$(m+1)$ 
can be written as follows.

\begin{align*}
&\frac{1}{n^{m+1}}\sum_{h\in \mathcal{H}}(1+\Index(m+1,h))\\
\intertext{By \Cref{lem:forest-properties}(3), $1+\Index(m+1) \leq |\Subtree(m+1)|$.}
&\leq  \frac{1}{n^{m+1}}\sum_{h\in\mathcal{H}} |\Subtree(m+1,F(h))|\\
\intertext{By \Cref{lem:numseq} we can consider only realizable trees, and count the number of hash functions in $\mathcal{H}$ consistent with each such tree.}
&\leq \frac{1}{n^{m+1}}\sum_{\substack{T \text{ is a realizable}\\ \text{tree rooted at } m+1}} |T|n(n-|T|)^{m+1-|T|}\\
\intertext{\Cref{lem:numgenerable} upper bounds the number of realizable trees on $s$ elements.}
&\leq \frac{1}{n^{m+1}}\sum_{s=1}^{m+1}s \binom{m}{s-1}\frac{4^{s-1}}{s^{3/2}}(s-1)!n(n-s)^{m+1-s}\\
\intertext{Canceling an $ns$, and reindexing from $s=0$, we have}
&\leq \sum_{s=0}^{m}\frac{1}{n^m}\frac{1}{\sqrt{s+1}}\binom{m}{s}4^{s}s!(n-s)^{m-s}.
\end{align*}

Define $A_s=\frac{1}{n^m}\frac{1}{\sqrt{s+1}}\binom{m}{s}4^{s}s!(n-s)^{m-s}$
to be the $s$th term of this sum.  
We may bound it as follows.

\begin{align*}
A_s
&=\frac{1}{n^m}\frac{1}{\sqrt{s+1}}\binom{m}{s}4^{s}s!(n-s)^{m-s}\\
\intertext{By Stirling's approximation $s! < e\sqrt{s}(s/e)^s$ for all $s$.}
&\leq \frac{e}{n^m}\binom{m}{s}\left(\frac{4s}{e}\right)^s(n-s)^{m-s}\\
\intertext{The probability of seeing $s$ successes in $m$ trials is ${m\choose s}p^s(1-p)^{m-s} \leq 1$, so for \emph{any} $p\in(0,1)$, ${m\choose s} \leq (p^s(1-p)^{m-s})^{-1}$.  We choose $p=\frac{s}{m}$.}
&\leq \frac{e}{n^m}\left (p^s(1-p)^{m-s}\right )^{-1}\left(\frac{4s}{e}\right)^s(n-s)^{m-s}\\
\intertext{Recall that $\alpha=m/n$ is the load factor,
so $n-s = n(1-\alpha p)$.}
&=\frac{e}{n^m}p^{-s}(1-p)^{s-m}\left(\frac{4s}{e}\right)^sn^{m-s}(1-\alpha p)^{m-s}\\
\intertext{Note that $\left(\frac{s}{pn}\right)^s = \alpha^s$ and $\frac{1-\alpha p}{1-p} = 1+\frac{p(1-\alpha)}{1-p}$.}
&=e\left(\frac{4\alpha}{e}\right)^s\left (\frac{1-\alpha p}{1-p}\right )^{m-s}\\
&\leq e\left(\frac{4\alpha}{e}\right)^se^{\frac{p(1-\alpha)}{1-p}(m-s)}\\
\intertext{and by definition of $p$, $p(m-s)/(1-p) = s$.}
&=e\left(\frac{4}{e}\alpha e^{1-\alpha}\right)^s.
\end{align*}

Recall that $\alpha^*\approx 0.357403$ is the unique 
solution to $\frac{4}{e}\alpha e^{1-\alpha} = 1$ 
in the range $[0,1]$.
Whenever $\alpha < \alpha^*$, 
the expected insertion cost is 
therefore 
$\sum_{s\geq 0} A_s = O(1)$.
\end{proof}

\section{Improved Load Factor for Quadratic Probing}\label{sect:quadratic-probing}

\Cref{lem:numgenerable} shows that the number of realizable labeled trees $T$ on $s$ elements 
is at most $4^{s-1}(s-1)!/s^{3/2} 
< (\frac{4}{e}s)^s$, 
but this bound is not tight.
In general, whenever we can prove the number
of realizable witness trees of size $s$
is $O((\beta s)^s)$, the analysis of \Cref{sect:counting-realizable-witness-trees}
implies the expected
insertion cost is constant whenever
the load factor is $\alpha<\hat{\alpha}$, 
where $\hat{\alpha}$ is the unique solution to 
$\beta\hat{\alpha} e^{1-\hat{\alpha}} = 1$ in $[0,1]$.

In this section we bound the number of 
witness trees for the quadratic probing sequence,
yielding $\beta \leq 1.42473$ 
and 
$\hat{\alpha} > 0.3761$.
Our method is to exactly calculate the number
of witness trees with $s$ vertices for constant values of $s=O(1)$, then leverage these values 
in a recurrence relation to asymptotically 
improve the count for all $s=\Omega(1)$.

\begin{theorem}\label{thm:mainthm2}
For any load factor $\alpha\leq 0.3761$, with probability $1-\exp(-\Omega(\sqrt{n}))$, the expected insertion time using quadratic probing is constant.
\end{theorem}

\Cref{lem:numquadratic} is the key technical 
lemma that implies \Cref{thm:mainthm2}.

\begin{lemma}\label{lem:numquadratic}
Let $r_i=i^2$ be the offset sequence, and $n$ be sufficiently large.
Fix any $V'\subseteq V$ with $|V'|=s$.
The number of realizable labeled witness trees on element set $V'$ 
is $O(\beta^ss^s)$, 
where $\beta\leq 1.42473$.
\end{lemma}

\begin{proof}
Rather than force the root $u$ of a witness tree $T$ to be labeled $\ell(u)=\bot$, 
suppose we let $\ell(u)$ be any 
descendant in $T$.
If we obtain an $O((\beta s)^s)$ 
bound on the number of such ``enhanced'' 
realizable witness trees,
it obviously extends to 
normal witness trees, 
divided by an $s$ factor.
Note that every subtree of an enhanced witness tree is a valid enhanced witness tree,
so we can count their number 
by dynamic programming.

Let $b_s$ be an \emph{upper bound} on the number of realizable enhanced trees with $s$ nodes that, without loss of generality, are named $1,2,\ldots,s$ and 
inserted in that order. 
In the base case $b_1=1$. For $s\geq 2$, we decompose the tree into the root and $t\geq 1$ subtrees having 
sizes $(s_0,s_1,\ldots,s_{t-1})$. The root was inserted last, so it is node $s$.
There are $\binom{s-1}{s_0,s_1,\ldots,s_{t-1}}$ ways to allocate nodes $1,\ldots,s-1$ to the root's subtrees. 
For each subtree, there are at most $b_{s_i}$ possible configurations. Summing over $t$ and $(s_0,s_1,\ldots,s_{t-1})$, we have:

\begin{align}
b_s 
&= s\sum_{t\geq 1}\sum_{s_0+s_1+\cdots +s_{t-1} = s-1,\, s_i\geq 1}\binom{s-1}{s_0,s_1,\ldots,s_{t-1}}\prod_{i=0}^{t-1} b_{s_i}.\nonumber\\
\intertext{Applying the identity 
$\binom{s-1}{s_0,s_1,\ldots,s_{t-1}} = \frac{(s-1)!}{s_0!s_1!\cdots s_{t-1}!}$, we have}
\frac{b_s}{s!}&=\sum_{t\geq 1}\sum_{s_0+s_1+\cdots +s_{t-1}=s-1,\, s_i \geq 1}\prod_{i=0}^{t-1}\frac{b_{s_i}}{s_i!}\label{eqn:rec-trees}
\intertext{Note that (\ref{eqn:rec-trees}) is 
the recurrence for $s\geq 2$, 
while there is just one tree when $s=1$.  
We shall express the recurrence of 
(\ref{eqn:rec-trees}) as a generating 
function 
$G(z) = \sum_{s\geq 1} \frac{b_s}{s!}z^s$.
Recall the standard notation that 
$[z^s]G(z) = b_s/s!$ is the coefficient of $z^s$
in the formal power series representation of $G$.
Being one node, the root contributes $z$ to shift the coefficient vector by 1,
while the choices for the children 
subtrees are captured by
$1 + G(z) + G(z)^2 + \cdots = \frac{1}{1-G(z)}$,
so}
    G(z) &= \frac{z}{1-G(z)}.\nonumber
\intertext{Solving $G(1-G)=z$ we have
$G(z) = \frac{1\pm \sqrt{1-4z}}{2}$.
The constant term in $G$ is 0, 
so the correct branch is:}
G(z)&=\frac{1-\sqrt{1-4z}}{2}.
\intertext{It is known~\cite[p.~738]{Flajolet09} 
that this is the generating 
function for the Catalan numbers $(C_{s-1})$, 
so we can also express $G$ equivalently as follows.}
G(z) &= \sum_{s\geq 1}C_{s-1}z^s.
\end{align}
Therefore it must be that $b_s=C_{s-1}s!$, which matches the results of \Cref{lem:numgenerable}.
However, this bound is not tight as there are witness trees for quadratic probing that 
satisfy~\Cref{lem:forest-properties} 
but are not realizable.  
See \Cref{fig:realizable} 
for one example.
The question is whether the number of realizable witness trees is \emph{significantly} smaller.

\medskip 

In a strict sense, the number of realizable witness trees 
of size $s$ in 
quadratic probing may depend on both $n$ and $s$,
but for $s < \sqrt{n}$ it is independent of $n$.
We shall therefore assume $n$ is 
sufficiently large.
Our strategy is to compute the exact number
of realizable enhanced witness trees for small 
values of $s=O(1)$, then use these in the recursive formulation to find asymptotically better bounds for large values of $s$.

Define $q_1, q_2, \ldots, q_k$ to be the 
exact number of realizable enhanced witness trees of sizes $1$ through $k$, and let $H(z)$ be an exponential generating function whose coefficients up to $k$ agree with $(q_i)_{1\leq i\leq k}$ and has the same structure as $G$ thereafter.
\begin{align}
[z^m]H(z) &=
\begin{cases}
    \dfrac{q_m}{m!}, & \text{when } m\leq k, \\
    [z^m]\dfrac{z}{1-H(z)},\istrut{7}    & \text{otherwise.}
\end{cases}
\intertext{Note that this definition is not circular.  The coefficient of $z^m$, $m>k$,
in $z(1+H(z)+H(z)^2 + \cdots)$ depends only on the previously defined coefficients up to degree $m-1$.
It will be useful to define some auxiliary functions.}
Q(z) &= \sum\limits_{m=1}^k \frac{q_m}{m!}z^m,\nonumber\\
R(z) &= \frac{z}{1-H(z)}\bmod z^{k+1},\nonumber\\
S(z) &= Q(z)-R(z).\nonumber
\end{align}

In other words, $Q$ reproduces just 
the first $k$ coefficients of $H$.
The coefficient $[z^m]R$, $m\leq k$, is 
what the recurrence (\ref{eqn:rec-trees}) 
\emph{would} give if each $b_{s_i}$ were replaced 
with the true count $q_{s_i}$.  The $\bmod z^{k+1}$ truncates all terms past $z^m$, $m>k$, in the power series.
$S$ is the difference between the 
true count and $R$'s upper bound on the true count.  Observe that $S$ also has only $k$ non-zero terms.
Since $[z^m]H(z)=[z^m]\frac{z}{1-H(z)}$ for $m>k$, we can express $H$ in this form, using $S$ to correct the first $k$ terms.

\begin{align*}
H(z) &=\frac{z}{1-H(z)}+S(z)\\
\intertext{and therefore}
0 &= H^2(z)-(1+S(z))H(z)+(S(z)+z)\istrut[3]{0}\\
H(z) &= \frac{1+S(z)\pm \sqrt{(1-S(z))^2 - 4z}}{2}.
\intertext{We have $S(0)=0$ and $H(0)=0$, so the correct branch is:}
H(z) &=\frac{1+S(z)-\sqrt{(1-S(z))^2-4z}}{2}.
\end{align*}

By the Exponential Growth Formula of 
analytic combinatorics~\cite[Page 244]{Flajolet09}, when $\rho$ is the absolute value of the singularity closest to $0$, $[z^m]H(z)\bowtie \left(\frac{1}{\rho}\right)^m$.\footnote{If $(a_m)$ is a sequence, $a_m \bowtie K^m$ is short for 
$\limsup_{m\to \infty} |a_m|^{1/m} = K$.} 
Let $D(z)=(1-S(z))^2-4z$.  
The only part of $H(z)$ that is potentially not analytic is $\sqrt{D(z)}$.
Every simple zero of $D$ is a square-root branch-point singularity. 
In this polynomial, the zero of minimum absolute value is simple.
\begin{equation}\label{eqn:rho}
\rho = \min\{|u| \::\: D(u)=0\neq D'(u)\}.
\end{equation}

In our calculations we took $k=15$, and 
enumerated all realizable enhanced witness trees with $s\leq 15$ vertices, yielding $q_1,\ldots,q_{15}$.\footnote{Note, for example, that there is one realizable witness tree with two nodes, since the root is unlabeled, but $q_2=2$ realizable \emph{enhanced} witness trees of size 2.}
\[
\begin{aligned}
   q_1 &=1 & \qquad\qquad q_6 &= 23970 & \qquad\qquad q_{11} &= 310866916602\\
   q_2 &=2 & q_7 &= 483588 & q_{12} &= 11478130014120\\
   q_3 &= 12 & q_8 &= 11336920 & q_{13} &= 461447509192826\\
   q_4 &= 108 & q_9 &= 305171856 & q_{14} &= 20062834712598860\\
   q_5 &= 1430 & q_{10} &= 9245198940 & q_{15} &= 938048043104586180.
\end{aligned}
\]
We explicitly generate the coefficients of $Q,R,S$, and calculate $\rho$ via (\ref{eqn:rho}).
The limiting bound
$[z^m]H(z)\bowtie \left(\frac{1}{\rho}\right)^m$ implies that for any $s$
and $\epsilon>0$,
$[z^s]H(z)=O((1/\rho+\epsilon)^s)$
and therefore the number of realizable
enhanced witness trees is
$O((1/\rho+\epsilon)^s s!)$,
and the number of realizable witness trees
is $O(\frac{1}{\sqrt{s}}(\frac{1}{e\rho}+\frac{\epsilon}{e})^s s^s)$.  We can take 
$\beta = 1.42473 > \frac{1}{e\rho}$.
Any $\alpha \leq 0.3761$ satisfies 
$\beta \alpha e^{1-\alpha}<1$.
This implies that quadratic probing enjoys constant expected insertion time for load factors up to $37.61\%$.\footnote{The code for generating $q_1,\ldots,q_{15}$ can be found at \url{https://github.com/gyh-20/Counting-Witness-Trees-for-Quadratic-Probing}.}
\end{proof}

\section{Conclusion}\label{sect:conclusion}

In this paper we presented a simple approach 
for analyzing quadratic probing and other 
fixed-offset hashing schemes, 
leading to results for some realistic load factors $>1/3$.  
The main open problem remains the same as it always was: to prove that quadratic probing
has constant expected insertion cost, ideally $O(\epsilon^{-1})$, at load factor 
$\alpha = 1-\epsilon$.

Our analysis introduces slack in three places, 
(1) upper bounding the true cost
$\Index(x)+1$ of inserting $x$
by the size of $\Subtree(x)$,
(2) overcounting the number of 
realizable witness trees,
and
(3) overestimating the probability of seeing a particular witness tree.
The effect of (1) seems to be modest,
and the method of \Cref{sect:quadratic-probing} mitigates the overcounting of unrealizable witness trees.
Our analysis of quadratic probing 
could be improved by improving \Cref{lem:numseq}'s
estimate of the probability of seeing a particular witness tree. Whereas \Cref{lem:numseq} depends only on the size $s$ of the tree, the ``true'' probability must also depend on the structure of the tree.

\bigskip 

\paragraph{AI Acknowledgment}
Some steps in the proofs were suggested by generative AI tools, namely:

\begin{itemize}
\item{Section 3}: 
The upper bound ${m\choose s} \leq (p^s(1-p)^{m-s})^{-1}$,
which holds for all $p\in (0,1)$.

\item{Section 4}: Introducing $Q(z),R(z),S(z)$ to help solve $H(z)$.

\item{Appendix A}: Applying the substitution $T(z)=-\log(1-F(z))$,
and applying Lagrange inversion to 
find $[z^n]F(z)$.
\end{itemize}

\bibliographystyle{alpha}
\bibliography{bibliography}

\appendix

\section{Witness Tree Analysis of Linear Probing}\label{sect:linear-probing}

Recall from \Cref{sect:witness,sect:quadratic-probing} 
that whenever we can bound the number of witness trees
by $O((\beta s)^s)$, this implies the expected
insertion cost is constant whenever
the load factor $\alpha$ satisfies $\alpha<\hat{\alpha}$, 
where $\hat{\alpha}$ is the unique solution to 
$\beta\hat{\alpha} e^{1-\hat{\alpha}} = 1$ in $[0,1]$. 

\medskip 

In this section we show the number of witness trees for linear probing is $(s-1)^{s-1}$, yielding
$\beta=1$ and $\hat{\alpha}=1$, that is,
the insertion cost is constant for any load factor $\alpha < 1$.

\begin{lemma}\label{lem:numlinear}
Consider the case of \emph{linear probing},
which uses offset sequence $r_i=i$.
For any element set $V'\subseteq V$ with $|V'|=s$, the number of realizable trees with elements and labels from $V'$ is $(s-1)^{s-1}$.
\end{lemma}

\begin{proof}
Consider a linear probing witness tree $T$ 
with root $u$ and children
$\child_0,\ldots,\child_k$.
Observe that for any $i\in [1,k]$ (excluding $i=0$),
we must have 
$\ell(\child_i) = \arg \min_{v\in \Subtree(\child_i)} \Pos(v)$. 
In linear probing, for any element $w$,
$S[w] = \{\Pos(v)\mid v\in \Subtree(w)\}$ 
is always a \emph{contiguous}
segment so the first position
of $S[\child_i]$ probed when inserting $u$
must be the \emph{first} position of the interval.
$S[\child_0]$ is an exception since 
$h(u)$ can be any position in $S[\child_0]$.
Let $\child_0(w)$ denote the first child of a non-leaf $w$.
Thus, we can uniquely identify $T$ by labeling
only $\child_0(w)$, for each $w\in T$.
For a fixed tree topology $\tau$ on element set $V'$,
the number of labelings is therefore
\begin{align*}
\prod_{w\in V', \text{ $w$ not a leaf of $\tau$}} |\Subtree(\child_0(w),\tau)|.
\end{align*}

Let $a_s$ be the number of realizable labeled 
witness trees on $V'=\{1,\ldots,s\}$, 
which were inserted in that order.
In the base case we have $a_1=1$. 
To derive a recurrence for $a_s$, 
we decompose the tree into the root and $t\geq 1$ subtrees
with sizes $s_0,s_1,\ldots,s_{t-1}$. 
The root must be $s$, while 
there are $\binom{s-1}{s_0,s_1,\ldots,s_{t-1}}$ choices 
for how to allocate the remaining elements to the subtrees.
For each subtree, there are $a_{s_i}$ possible configurations,
and for the first subtree we have $s_0$ choices for its label.
Summing over $t$ and $(s_0,s_1,\ldots,s_{t-1})$, 
we have:

\begin{align*}
a_s = \sum_{t\geq 1}\sum_{s_0+s_1+\cdots +s_{t-1}=s-1,\, 
s_i\geq 1}\binom{s-1}{s_0,s_1,\ldots,s_{t-1}}s_0\prod_{i=0}^{t-1} a_{s_i}
\end{align*}

Replacing $\binom{s-1}{s_0,s_1,\ldots,s_{t-1}}$ with 
$\frac{(s-1)!}{s_0!s_1!\cdots s_{t-1}!}$, we have:
\begin{align}\label{eqn:ar-recurrence}
s\frac{a_s}{s!}=\sum_{t\geq 1}\sum_{s_0+s_1+\cdots +s_{t-1}=s-1,\, s_i\geq 1}s_0\prod_{i=0}^{t-1}\frac{a_{s_i}}{s_i!}
\end{align}

Define $F(z)=\sum\limits_{s\geq 1}\frac{a_s}{s!}z^s$ to be 
the exponential generating function, and recall
that whenever $G$ is a formal power series in $z$,
$[z^m]G$ denotes the coefficient of $z^m$ in $G$.
We can now rewrite the recurrence above in terms of $F$. 
Observe that $F'(z)=\sum_{s\geq 1}\frac{a_s}{(s-1)!}z^{s-1}$,
so the left hand side of (\ref{eqn:ar-recurrence}) 
is $[z^s]zF'(z)$.  
The right side is multiplying the generating function for the first (labeled) child, which is $zF'(z)$, 
and the generating function for the sequence of 
zero or more remaining children,
namely $\frac{1}{1-F(z)} = 1+F(z)+F(z)^2+\cdots$.
We assumed $s>1$, 
so the subtrees contribute $s-1\geq 1$ nodes.
There is a single realizable tree when $s=1$.
To account for the root node we shift the coefficient vector by 1 by multiplying by $z$,
hence:

\begin{align}
zF'(z) &= z\left(1+\frac{zF'(z)}{1-F(z)}\right)\nonumber
\intertext{and therefore}
F'(z) &=1+\frac{zF'(z)}{1-F(z)}.\label{eqn:Fprime}
\end{align}

To solve this differential equation, define 
$T(z)=-\log(1-F(z))$, 
so $T'(z)=\frac{F'(z)}{1-F(z)}$,
$F(z)=1-e^{-T(z)}$, 
and
$F'(z)=T'(z)e^{-T(z)}$. 
Equation (\ref{eqn:Fprime}) can now be expressed as:
\begin{align}
T'(z)e^{-T(z)}&=1+zT'(z)\nonumber
\intertext{and therefore}
T'(z) &=e^{T(z)}+zT'(z)e^{T(z)}.\label{eqn:Tprime}
\intertext{The constant term of $T(z)$ is 0.
Observe that $\frac{\mathrm{d}}{\mathrm{d}z}\left(ze^{T(z)}\right)=e^{T(z)}+zT'(z)e^{T(z)}$.
We integrate (\ref{eqn:Tprime}) on both sides and have}
T(z) &= ze^{T(z)}.
\end{align}

$T(z)$ is the Cayley tree function, the generating function 
for rooted unordered labeled trees, 
whose power-series representation is 
$T(z)=\sum\limits_{s\geq 1}\frac{s^{s-1}}{s!}z^s$.\footnote{There are $s^{s-2}$ labeled trees
on $s$ nodes, and $s$ choices for the root, hence $s^{s-1}$ rooted labeled trees.} 
Defining $\phi(u)=ue^{-u}$, we have 
$\phi(T(z))=ze^{T(z)}e^{-T(z)}=z$
and hence $T(z) = \phi^{-1}(z)$. Recall the Lagrange inversion theorem:
\begin{align*}
[z^m]H(T(z))=\frac{1}{m}[z^{m-1}]H'(z)\left(\frac{z}{\phi(z)}\right)^m.
\end{align*}

We have $F(z)=1-e^{-T(z)}$, so take $H(z)=1-e^{-z}$ and $F(z)=H(T(z))$, apply the Lagrange inversion theorem:

\begin{align*}
a_s&=s![z^s](1-e^{-T(z)})\\
&=s!\frac{1}{s}[z^{s-1}](1-e^{-z})'\left(\frac{z}{\phi(z)}\right)^s\\
&=s!\frac{1}{s}[z^{s-1}]e^{(s-1)z}\\
&=s!\frac{1}{s}\frac{(s-1)^{s-1}}{(s-1)!}\\
&=(s-1)^{s-1}.
\end{align*}
\end{proof}

\begin{theorem}\label{thm:linear-probing}
    With load factor 
    $\alpha < 1$ bounded away from 1, the expected insertion cost for 
    hashing with 
    linear probing is constant.
\end{theorem}
\end{document}